\documentclass{article}
\usepackage[left=2.5cm, right=1.5cm, top=3cm, bottom=2cm]{geometry}
\usepackage[T1]{fontenc}
\usepackage[utf8]{inputenc}
\usepackage[english]{babel}
\usepackage{csquotes} % Recommended by biblatex
\usepackage{amsmath}
\usepackage{amsthm}
\usepackage{amssymb}
\usepackage{fancyhdr}
\usepackage{setspace}
\usepackage{graphicx}
\usepackage[dvipsnames]{xcolor} % Load xcolor early with options
\usepackage{colortbl}
\usepackage{tikz}
\usepackage{pgf}
\usepackage{subcaption}
\usepackage{listings}
\usepackage{indentfirst}
\usepackage[
backend=bibtex,
style=numeric,
sorting=none,
maxbibnames=99
]{biblatex}
\bibliography{refs}

\usepackage[colorlinks,citecolor=blue,linkcolor=blue,bookmarks=false,hypertexnames=true, urlcolor=blue]{hyperref} 
\usepackage{mathtools}
\usepackage{booktabs}
\usepackage[flushleft]{threeparttable}
\usepackage{tablefootnote}
\usepackage{pst-solides3d}
\usepackage{xargs}                      % Use more than one optional parameter in a new commands
\usepackage{pdfpages}

\usepackage{chngcntr} % numbering of figures and tables by sections
\counterwithin{table}{section}
\counterwithin{figure}{section}

\graphicspath{{graphics/}}

\colorlet{darkred}{red!80!black}
\colorlet{darkblue}{blue!80!black}
\colorlet{darkgreen}{green!60!black}

\usepackage{secdot}
\sectiondot{subsection}
\usepackage{verbatim}
\usepackage{float}
\usepackage{todonotes}
\usepackage{enumerate}

\title{On the Complexity of Recurrence Evaluation}

\author{Artem Parfenov\thanks{
National Research University Higher School of Economics, Moscow, Russia; e-mail: dunno\_o@icloud.com}
\and
Michael Vyalyi\thanks
{Federal Research Center ``Computer Science and Control'' of the
  Russian Academy of Sciences, Moscow, Russia; National Research University Higher School of Economics, Moscow, Russia;  Institute of Physics and Technology, Dolgorpudnyi, Russia; 
 e-mail:
vyalyi@gmail.com}
}

\date{May 31, 2026}

\newtheorem{theorem}{Theorem}
\newtheorem*{theoremSTAR}{Theorem}

\newtheorem{lemma}{Lemma}

\newtheorem{claim}{Claim}

\theoremstyle{definition}
\newtheorem{defi}{Definition} 
\newtheorem{rem}{Remark} 

\newtheorem{prob}{Problem}

\let\es\varnothing

\def\bin{\mathop{\mathrm{bin}}\nolimits}
\newcommand{\NN}{\mathbb{N}}
\newcommand{\ZZ}{\mathbb{Z}}

\let\wbeg\triangleleft
\let\wend\triangleright

\newcommand{\D}{\mathcal{D}}

\def\cA{\mathcal A}

\newcommand{\blankTM}{\text{\texttt{\textvisiblespace}}}
\newcommandx{\unsure}[2][1=]{\todo[linecolor=red,backgroundcolor=red!25,bordercolor=red,#1]{#2}}
\newcommandx{\change}[2][1=]{\todo[linecolor=blue,backgroundcolor=blue!25,bordercolor=blue,#1]{#2}}

\newcommand{\poly}{\ensuremath{\mathrm{poly}}}
\newcommand{\CGE}{\ensuremath{\mathrm{CGE}}}
\newcommand{\sCGE}{\ensuremath{\mathrm{sCGE}}}
\newcommand{\symF}{\ensuremath{\mathrm{SymBF}}}

\newcommand{\PSPACE}{\ensuremath{\mathsf{PSPACE}}}
\newcommand{\Ptime}{\ensuremath{\mathsf{P}}}
\newcommand{\DTIME}{\ensuremath{\mathsf{DTIME}}}
\newcommand{\EXP}{\ensuremath{\mathsf{EXP}}}
\newcommand{\PP}{\ensuremath{\mathsf{PP}}}
\newcommand{\TC}{\ensuremath{\mathsf{TC}}}
\newcommand{\DLOGTIME}{\ensuremath{\mathsf{DLOGTIME}}}
\newcommand{\SIZE}{\ensuremath{\mathsf{SIZE}}}
\newcommand{\plusP}{\ensuremath{{\oplus}\mathsf{P}}}

\usetikzlibrary{calc, patterns}

\begin{document}
    \maketitle
	
	{
		\hypersetup{linkcolor=black}
		%\tableofcontents
	}

    \begin{abstract} 
      In this paper, we study the complexity of the recurrence evaluation problem. We are interested in finitely valued recurrent functions. We present two results in this direction. First, we study the recurrence problem for
      sequences, assuming that a recurrence relation is defined by a fixed function, while the offsets are part of the input. Depending on the form of presentation (whether the offsets are given in unary or in binary), the problem is \PSPACE-complete or \EXP-complete.

Second, we study recurrences defined by the NAND function. They are related to impartial games. 
We prove \PP-hardness of the recurrence evaluation problem for a very simple $3$-dimensional game, in which the offset vectors  are coordinate vectors $(1,0,0)$, $(0,1,0)$ and $(0,0,1)$ but  the boundary conditions  are arbitrary. In other words, we consider generalized  winning conditions for the game extending the normal and the mis\`ere winning conditions.

\textbf{Keywords:}
      Recurrences,  impartial games, space complexity, time complexity.  

\end{abstract}

	\section{Introduction}\label{sec:introduction}
	
	In many mathematical problems, functions are defined inductively:  the next value depends on previous ones. In theoretical computer science, a key question for any problem is the algorithmic complexity of its solution. In this paper we address this question for recurrent functions with domain $\NN^k$, an integer $k$-dimensional orthant, that take values in a finite set. A~function of this sort is defined by a  recurrence
\begin{equation}\label{basic-rec}
P(x) = f(P(x-a_1), P(x-a_2), \dots, P(x-a_d)),\quad x\in \NN^k,
\ a_i\in \ZZ^k, \ \sum_ia_i>0.
\end{equation}
It is a discrete version of a differential equation. If $f$ is a linear function, the recurrence is called \emph{linear}. The linear recurrences were intensively studied and have numerous applications including error-correcting codes. See the book \cite{EPShW03} for the detailed exposition.

It is well-known that, in the case of finitely valued recurrent sequences ($k=1$), a~recurrent sequence is eventually periodic. Indeed, the right hand side of the recurrent relation~\eqref{basic-rec} depends on a bounded number of previous terms, say $D = \max_{i} a_i$. Since the sequence takes values in a finite set, there are only finitely many possible tuples of $D$ consecutive values. By the Pigeonhole Principle, some $D$-tuple must eventually repeat. Because the recurrence is deterministic, the sequence becomes periodic from that point onward. It implies that evaluating this sequence belongs to the complexity class \Ptime. But a period of a finitely valued recurrent sequence can be exponentially large in the size of its description. It makes evaluation of a  recurrence  non-trivial if the~recurrence itself is a part of the input. To the best of our knowledge, the complexity of such a problem  is open. Note  that for linear recurrences over a finite field the period can be exponentially large. Nevertheless,  there exists a uniform polynomial time evaluation algorithm based on linear algebra (see, e.g.,~\cite{DBLP:journals/cj/MillerB66} or \cite {F85, Kh18} for more efficient versions).

For larger dimensions, the situation is drastically different. There exists a fixed finitely valued recurrence in the form~\eqref{basic-rec} such that the evaluation problem for this particular sequence is \EXP-complete even for $k=2$. This fact can be easily proved using the relation between 1-dimensional cellular automata and Turing machines.    

We are mainly interested in recurrences related to impartial games. In this case, the values of a recurrence are P and N (corresponding to a winner in a position). We will assign 0 to P-positions and 1 to N-positions. In that form, a function in a recurrence defined by impartial game rules is the NAND function. NAND-recurrent functions defined by recurrences in the form~\eqref{basic-rec} correspond to (multidimensional) subtraction games.

Ordinary subtraction games for $k=1$ were introduced in Golomb's paper \cite{GOLOMB1966443}. Note that winning conditions (say, normal or mis\`ere) determine the boundary conditions for the NAND-recurrence. For normal and mis\`ere winning conditions, the questions about period lengths and algorithmic complexity of NAND-recurrences are largely open.
The question about periods originates from~\cite{Guy1995}. For normal winning conditions, in~\cite{ALTHOFER1995111} superlinear lower bounds for period lengths were found.
Results for a fixed set of differences are given in \cite{HO201535,ZHANGLinPerSubGames}.

To date, there are no superpolynomial lower bounds for either normal or mis\`ere winning conditions.
For general boundary conditions, in a recent paper~\cite{DBLP:journals/ijgt/MiklosP24}, superpolynomial lower bounds were proved. The algorithmic complexity of evaluating NAND-recurrences is still open even for general boundary conditions.

For multidimensional subtraction games, Larsson proved in \cite{larsson2012heapsmatcheslimitscomputability} that, for some dimension $k$, the problem of determining the equivalence of two games is undecidable, which means no algorithm exists to solve this task.
In~\cite{VyaliyMultidimensionalSubtractionGames} this result was extended to algorithmic complexity. For some dimension $k$, there exists a multidimensional subtraction game such that solving the game is \EXP-complete.

Gurvich and Vyalyi conjectured in~\cite{VyaliyMultidimensionalSubtractionGames} that for non-negative offsets (vectors $a_i$ in~\eqref{basic-rec}), the set of P-positions is semilinear. It would imply that solving such games takes polynomial time. In specific cases (slow Nim($n,k$) for $n=6$, $k=2$ and $n=k+1$), the semilinear conjecture holds, see~\cite{Ch-etc21,Getc24}.  
    In the paper \cite{larsson2024subtractiongamesdimension}, a possible counterexample to the semilinear conjecture was presented. Computational experiments reveal a rather chaotic behaviour of P-positions of the proposed game, although a formal proof of non-semilinearity was not presented.

\paragraph{Our contribution.}
    
In this paper, we present new results in two directions. First, we prove the computational hardness of evaluating a 1-dimensional recurrent sequence for a fixed function $f$. The offsets $a_i$ are part of the input. Depending on the form of presentation (whether the offsets $a_i$ are given in unary or in binary), the problem is \PSPACE-complete or \EXP-complete. The function $f$ in the proof is related to a universal Turing machine and is far from the NAND-function. Nevertheless, this result provides evidence in favor of the computational hardness of solving subtraction games.

For NAND-functions, we prove the computational hardness for a very simple $3$-dimensional game, in which the offset vectors $a_i$ are coordinate vectors $(1,0,0)$, $(0,1,0)$, and $(0,0,1)$. The hardness stems from the boundary conditions, which are arbitrary (for normal or mis\`ere winning conditions, the game is trivial).
The evaluation problem for this NAND-recurrence is \PP-hard. This result shows that mixed winning conditions are much harder than normal or mis\`ere ones, and can be regarded as an algorithmic version of the main result in~\cite{DBLP:journals/ijgt/MiklosP24}.

        \section{Definitions and notation}\label{sec:defs}

We assume that a reader is familiar with computational complexity theory on a basic level. So we  
 skip the basic definitions of the theory: reductions, complexity classes, etc. They can be found in textbooks \cite{sipser2012introduction, AroraBorak, DuKo}. We make a technical remark about Turing machines (TM). We assume that the tape of a Turing machine is infinite in one direction and the input word starts from the leftmost position of the tape.

 For a word $w = w_0w_1\dots w_{n-1}$, we denote by $w[a:b]$ the subword (the segment) $w_a\dots w_b$.

We adopt the notation $\ZZ$ for integers and $\NN$ for non-negative integers.
For vectors $x,y\in \ZZ^k$, we use the notation $x\geq y$ for the coordinate-wise order: $x\geq y$ is equivalent to $x_i\geq y_i$ for all $1\leq i\leq k$.

Let $\Sigma$ be a finite set (an alphabet), $f\colon \Sigma^d \to \Sigma$ be a Boolean function, $\cA =\{a_1, \dots, a_d\}$, $a_i>0$, be a set of vectors in $\NN^k$ called \emph{offsets}, and $\D\subseteq \ZZ^k$ be a region in $\ZZ^k$. 
A function $P\colon \D\to \Sigma$ is an $(f,\cA)$-\emph{recurrent function} if 
\begin{equation}\label{general-rec}
P(x) = f(P(x - a_1), P(x - a_2), \dots, P(x - a_d)) 
\end{equation}
holds for any $x\in\D$ such that $x-a_i\in \D$ for all $1\leq i \leq d$.

To identify an $(f,\cA)$-recurrent function $P$ on a region $\D$,
one needs to specify boundary conditions, i.e., to determine the values of $P$ in those region points, for which the right hand side of~\eqref{general-rec} is undefined. We refer to these points as  \emph{boundary points}.

In this paper, we are interested in algorithmic complexity of solving recurrences.
More specifically, our base problem is to compute $P(x)$ for an $(f,\cA)$-recurrent function $P$ on a region~$\D$. The function $f$ is a parameter of the problem. Thus, $d=O(1)$. Therefore, we always assume that $f$ is given by the table of values. %The size of the table is $|\Sigma|^d$. 
In general, all other data are part of an input.

\begin{defi}
Let $\Sigma$ be a finite set, $f\colon \Sigma^d \to \Sigma$ be a Boolean function, $\cA =\{a_1, \dots, a_d\}$, $a_i>0$, be a set of positive integers.
An $(f,\cA)$-recurrent function on the region $\NN$ is called an $(f,\cA)$-\emph{recurrent sequence}.

If $\Sigma=\{0,1\}$, then an $(f,\cA)$-recurrent sequence is called \emph{Boolean}. 
\end{defi}

Specifying our base problem to recurrent sequences, we obtain the following two problems. 

\begin{prob}[\boldmath $RE (f,\Sigma,b)$ Recurrent Sequence Evaluation]\phantom{}
		Input: 
        the set of the offsets $\cA=\{a_1,\dots, a_d\}$ of an $(f,\cA)$-recurrent sequence are given in unary; the boundary conditions $P[0:n-1]\in \Sigma^n$, where $n = \max_{1\leq i\leq d} a_i$; the index $x$ is given in binary representation.
        
		Question: Is it true that $P(x) = b$?
	\end{prob}

If the offsets are given in binary, then there are exponentially many boundary points. To specify an input recurrent sequence, we need a succinct representation of a sequence  $P[0:n-1]$ of exponential length. We adopt the following way to represent this set succinctly. For later use we fix a special injection $\tau\colon \Sigma\to\{0,1\}^{\ell}$ defined as follows.

	\begin{defi}[Encoding $\tau$]\label{CodingDef}
		Map each symbol $s_i \in \Sigma = \{s_1, \dots s_k\}$ to a word $\tau(s_i)$ in the alphabet $\{0,1\}$ according to the following rule: $s_i \to 1100w_i0011$, $w_i = 0^{i-1}10^{k-i}$. 
		The code length of a symbol is $l = k + 8$.
        
    A word $s = s_0 s_1 \dots s_{n-1} \in\Sigma^*$ is encoded by a binary word $\tau(s) = \tau(s_0) \tau(s_1) \dots \tau(s_{n-1})$.
    \end{defi}

  A Boolean circuit $C\colon\{0,1\}^{\lceil \log n\rceil}\to \{0,1\}^\ell$ represents the sequence $P[0:n-1]$ if $C(\bin(y)) = \tau(P(y)) $, where $\bin(y)$ represents $y$ in binary and has length $\lceil \log n\rceil$.

	\begin{prob}[\boldmath $SRE (f,\Sigma,b)$ Succinct Recurrent Sequence Evaluation]\phantom{}

		Input: the  set of the offsets $\cA=\{a_1,\dots, a_d\}$ of an $(f,\cA)$-recurrent sequence are given in binary; a~Boolean circuit $C$ representing the boundary conditions (the sequence $P[0:n-1]$); the index $x$ is given in binary representation.
        
		Question: Is it true that $P(x) = b$?
	\end{prob}

	\section{Complexity of Evaluating Recurrent Sequences}

The straightforward way to compute the value of the $i$-th member of a recurrent sequence is to compute all preceding members starting from the boundary points and using a recurrence relation.
	Our goal is to show that this recursive algorithm is essentially the most efficient for solving general recurrences. More formally, we are going to prove the following claims.

	\begin{theorem}\label{result}
		There exists a Boolean function $f$ such that the problem $RE (f,\Sigma,b)$  is \PSPACE-complete. Moreover, there exists a Boolean function $g$ such that $RE (g,\{0,1\},b)$ is  \PSPACE-complete. 
	\end{theorem}

	\begin{theorem}\label{result2}
		There exists a Boolean function $f$ such that the problem $SRE (f,\Sigma,b)$  is $\EXP$-complete. Moreover, 
        there exists a Boolean function $g$ such that $RE (g,\{0,1\},b)$ is \EXP-complete. 
	\end{theorem}

\subsection{Versions of the Bounded Halting Problem}

The proofs of Theorems~\ref{result} and \ref{result2} are based on reductions. We will reduce to $(S)RE$ problems the variants of the bounded halting problem.
    
	\begin{theoremSTAR}[Space Bounded Halting Problem (SBHP)]\label{Halting} \cite[Theorem~3.23]{DuKo}
		
		The language $\mathrm{SBHP}$ consists of triples $\langle M, w, 1^s \rangle$ such that the deterministic TM $M$ accepts the input $w$ in space~$s$.
		
		This language is complete in the class \PSPACE{}.
	\end{theoremSTAR}
	
	\begin{theoremSTAR}[Exponential-Time Bounded Halting Problem (EXP-BHP)]\label{ExpHalting} \cite[Proposition~3.30]{DuKo}
		
		The language $\mathrm{EXP\text{-}BHP}$ consists of triples $\langle M, w, n \rangle$ where $n$ is written in binary, such that the deterministic TM $M$ halts on input $w$ in $n$ moves.
		
		This language is complete in the class \EXP.      
	\end{theoremSTAR}

        Let $M_{\PSPACE}$ be  a fixed deterministic Turing machine that decides a \PSPACE-complete language in polynomial space. For this fixed machine, the problem of deciding whether $M_{\PSPACE}$ accepts an input $w$ within space $s$ (given in unary) remains \PSPACE-complete. Similarly, there exists a  deterministic Turing machine $M_{\EXP}$ deciding an \EXP-complete language, for which the problem of deciding whether $M_{\EXP}$ halts on $w$ in $n$ steps (given in binary) is \EXP-complete.

	In the proofs below we need a special form of Turing machines.

\begin{defi}
    	The \emph{bordered} TM  $M= (Q, \Gamma, \sigma, q_0, q_f)$  satisfies the following conditions:
		\begin{enumerate}
			\item The alphabet $\Gamma$ contains special delimiter symbols $\wbeg, \wend\in\Gamma$ such that  
            \[
            \sigma(q, \wbeg) = (q', \wbeg, 1)\ \text{and}\  
            \sigma(q, \wend) = (q', \wend, -1)\]
            for all $q\in Q$, i.e., the machine does not move outside the workspace between the delimiters.
			\item The machine $M$ has exactly two special states: a ``preliminary'' final state $q_r$ and an actual final state $q_f$, such that $\sigma(q_r, \alpha) = (q_r, \alpha, +1)$ for all $\alpha \neq \wend$ and $\sigma(q_r, \wend) = (q_f, \wend, 0)$.
		\end{enumerate}
\end{defi}

We define a family of halting problems using a fixed bordered TM $M$ as a parameter.

	\begin{prob}[Bordered $\mathrm{SBHP}_M$]\label{NewHaltingSpace}\phantom{}
		Input: the input word $w$, the space size $S$ in unary representation.
		
		Question: Does the bordered TM $M$ halt starting from the 
        initial configuration: $\wbeg q_0 w \underbrace{\blankTM \dots \blankTM}_{S-|w|} \wend$, where  $q_0$ is the initial state of the machine, and $\blankTM\in\Gamma$ is the blank symbol?
	\end{prob}

	\begin{prob}[Bordered $\mathrm{EXP\text{-}BHP}_M$]\label{NewHaltingTime}\phantom{}
		Input: the input word $w$, the time bound $n$ in binary representation.
		
		Question: Does the bordered TM $M$ halt starting from the 
        initial configuration: $\wbeg q_0 w \underbrace{\blankTM \dots \blankTM}_{n-|w|} \wend$, where  $q_0$ is the initial state of the machine, and $\blankTM\in\Gamma$ is the blank symbol?
	\end{prob}

	\begin{lemma}\label{lem:bordered_pspace}
	  There exists a fixed bordered TM $M$ such that the Bordered $\mathrm{SBHP}_M$  is \PSPACE-complete.

   There exists a fixed bordered TM $M$ such that the Bordered  $\mathrm{EXP\text{-}BHP}_M$ is \EXP-complete.    
	\end{lemma}
	
	\begin{proof}
	    Let $M_{\PSPACE}= (Q, \Gamma, \sigma, q_0, q_y,q_n)$ be a fixed deterministic Turing machine that decides a \PSPACE-complete language in  space $p(|w|)$, where $w$ is the input word and $p(\cdot)$ is a polynomial. The problem of deciding whether $M_{\PSPACE}$ accepts $w$ in space $s$ (where $s \ge p(|w|)$ is given in unary) is \PSPACE-complete.
	    
	    We construct $M$ as a bordered modification of $M_{\PSPACE}$. The states of $M$ are the states of $M_{\PSPACE}$ and two additional states $q_r$,  $q_f$. The transition function $\sigma$ coincides with the transition function of $M_{\PSPACE}$ except the following rules that guarantee that the conditions of a bordered machine hold:
            \[
        \begin{aligned}
            &\sigma(q, \wbeg) = (q_n, \wbeg, 1);\\
            &\sigma(q, \wend) = (q_n, \wend, -1);\\
            &\sigma(q_n, \alpha) = (q_n, \alpha,0) &&\text{for all $\alpha\in \Gamma$};\\
            &\sigma(q_y, \alpha) = (q_r, \alpha,0) &&\text{for all $\alpha\in \Gamma$};\\
            &\sigma(q_r, \alpha) = (q_r, \alpha, +1) &&\text{for all $\alpha \in \Gamma\setminus\{\wend\}$};
            &\sigma(q_r, \wend) = (q_f, \wend, 0).
        \end{aligned}
        \]
        The parameter $S$ is set to $s$. The machine $M$ starts from the initial configuration $\wbeg q_0 w\underbrace{\blankTM \dots \blankTM}_{S-|w|} \wend$ and simulates $M_{\PSPACE}$.
	    
	    If $M_{\PSPACE}$ accepts the input, $M$ enters the ``preliminary'' final state $q_r$, moves its head to the right until it reaches $\wend$, and halts in $q_f$. If $M_{\PSPACE}$ rejects the input or attempts to go outside the workspace, it enters an infinite loop keeping the  absorbing state $q_n$. 
        
	    This transformation takes polynomial time because we only need to compute $S = p(|w|)$ in unary and construct the initial configuration. Since $M_{\PSPACE}$ accepts $w$ in space $s$ if and only if the fixed bordered machine $M$ halts on space $S=s$, 
	    the bordered $\mathrm{SBHP}_M$ problem is \PSPACE-complete.

The second claim is proved similarly. The machine $M_{\PSPACE}$ should be replaced by a machine $M_{\EXP}$ deciding an \EXP-complete language in time $2^{p(|w|)}$, where $p(\cdot)$ is a polynomial. In this case the parameter $n$ is $2^{p(|w|)}$, where $w$ is the input word.  
	\end{proof}

\subsection{Reductions to Binary Alphabet}

	The encoding $\tau$ introduced in Section~\ref{sec:defs} admits the following alignment property.

	\begin{lemma}\label{CodingLemma}
		Let $\tau(s) = \tau(s_0) \tau(s_1) \dots \tau(s_{n-1})$ be an encoding of a word $s=s_0\dots s_{n-1} \in \Sigma^*$.  For any position $x\geq l+2$, the segment $\tau(s)[x - l - 2: x]$ determines uniquely the position $y$ in the symbol code $\tau(s_m)$, where the position $x$ is a position inside $\tau(s_m)$ in the word $\tau(s)$.
	\end{lemma}
	
	\begin{proof}
		Our goal is to restore alignment, i.e., find the position of the first bit of the symbol code $\tau(s_m)$. The position $y$ is determined by this information. Let us present an aligning algorithm. Note that in the code of any word, four consecutive ones occur only at the boundaries of two symbol codes.
        Find the maximal $e$ such that the segment $\tau(s)[e: e+3]$ contains only 1s and $e+3 \leq x$. Informally, the segment is         
        the nearest quadruple of 1s to the left of $x$. Then the symbol code starts at position $e + 2$.
        
		By definition, $e$ is the starting position of the last quadruple of 1s that fully fits before $x$ (i.e., $e+3 \le x$). The next quadruple of 1s will start exactly $l$ symbols later, at position $e+l$, and end at $e+l+3$. Since $e$ is maximal, this next quadruple cannot fully fit into the prefix up to $x$. This means $e + l + 3 > x$, which implies $e \ge x - l - 2$. Thus, the minimal possible value of $e$ is $x - l - 2$. Therefore, it is sufficient to look at the segment $\tau(s)[x - l - 2: x]$, which has length $x - (x - l - 2) + 1 = l + 3$.
	\end{proof}

		We  reduce the recurrence evaluation problem in a fixed finite alphabet  to the binary recurrence evaluation problem in the following way. Let  $P\colon \NN \to \Sigma$, where $|\Sigma| < \infty$, be a recurrence 
		\begin{equation}\label{rec-eq}
	P(x) = f(P(x - a_1), P(x - a_2), \dots, P(x - a_d)), \quad P(0)= P_0,\ P(1) =P_1 ,\ \dots, P(a-1) = P_{a-1}, \ a = \max_i a_i. 	    
		\end{equation}

		We want to define a binary sequence $T = \tau(P(0)) \tau(P(1)) \dots $ by a recurrence with the maximum offset $l \cdot a$, where $l= |\Sigma|+8$. Boundary conditions are given by the $\tau$-encoding of the first $a$ members of the sequence $P$. Define a Boolean function $g$ depending on $d \cdot l + (l + 3)$ variables. Informally, the arguments of this function are $\tau$-codes of the arguments of function $f$ and
        $l+3$ members of the sequence preceding the current position. Accordingly, the required offsets from the current position $i$ are of two types: offsets $1, 2, \dots, l+3$ to restore the alignment, and offsets of the form $j + l \cdot a_v - r$, where $0 \leq r \leq l - 1$, $1\leq v\leq d$, to read the encodings of the function $f$ arguments. Here $j$ ($0 \leq j \leq l - 1$) is the index of the current bit within the code of $P(x)$. From the segment $T[i - l - 3: i - 1]$, one can determine the start position $i - j$ of the code of the current symbol of sequence $P$ according to Lemma \ref{CodingLemma}. The value of the current symbol is determined by function $f$ with arguments $P(x - a_1), P(x - a_2), \dots, P(x - a_d)$. The code of symbol $P(x - a_v)$ fills the segment $$[i - j - a_v \cdot l, i - j - a_v \cdot l + l - 1].$$ Function $g$ determines symbols $P(x - a_1), P(x - a_2), \dots, P(x - a_d)$ from their encodings, applies function $f$ to them, and outputs the value equal to the $j$-th bit of the code $f(P(x - a_1), P(x - a_2), \dots, P(x - a_d))$.
		
		The value table of function $g$ has size $2^{d \cdot l + (l + 3)}$, which is polynomially bounded by the size of the value table of function $f$, i.e., $|\Sigma|^d$, provided $l = |\Sigma| + 8 = O(1)$.

These arguments can be applied to prove  the following lemmas.

\begin{lemma}\label{Sigma2bin}
For any Boolean function $f$ there exists a Boolean function $g$ such that 
    $RE (f,\Sigma, b)$ is polynomially reducible to $RE (g,\{0,1\}, 1)$. 
\end{lemma}

\begin{proof}
    The function $g$ defines the recurrence for the binary sequence $T$, which is  the sequence of $\tau$-codes of the symbols in the sequence defined by the recurrence~\eqref{rec-eq}. Let $x$ be the index in the input of $RE (f,\Sigma, b)$ and $k$ be the number of the symbol $b$ in the enumeration of $\Sigma $ used in $\tau$-encoding. Then the index in the input of $RE (g,\{0,1\}, 1)$ is $(x-1)(|\Sigma|+8)+4+k$. It is computed in polynomial time as well as offsets for the recurrence defining $T$ and initial data of this sequence. 
\end{proof}

\begin{lemma}\label{SuccSigma2bin}
For any Boolean function $f$ there exists a Boolean function $g$ such that 
    $SRE (f,\Sigma, b)$ is polynomially reducible to $SRE (g,\{0,1\}, 1)$. 
\end{lemma}

The proof is the same as for Lemma~\ref{Sigma2bin}. Note that the Boolean circuit $C$ is the same since we use $\tau$-encoding for the sequence $T$.

	\subsection{Proof of Theorem \ref{result}}

        To prove Theorem \ref{result}, we need two results. First, we note that the problem $RE (f,\Sigma, b)$ is solvable in polynomial space for any $f, \Sigma, b$. To compute the recurrence in polynomial space, we need to maintain an array of $\max_i a_i$ elements and a counter. Since the next value depends only on previous ones no further back than $\max_i a_i$, we only need to store them in the array and update the array  and the counter. Since the offsets are given  in unary, the memory used is polynomially bounded by the input size. 

To prove \PSPACE-hardness, we apply Lemma~\ref{lem:bordered_pspace} and choose  a bordered TM 
 $M$ such that the bordered $\mathrm{SBHP}_M$ (Problem~\ref{NewHaltingSpace}) is \PSPACE-complete. After that, we reduce the bordered $\mathrm{SBHP}_M$ to $RE (f,\Sigma, b)$ for appropriately chosen $f$, $\Sigma$, $b$ by a polynomial $m$-reduction.

The main idea is to define a recurrence generating a \emph{protocol}:  a sequence  of consecutive configurations of the TM. We assume that a halting configuration is infinitely repeated in the protocol.  Since the machine is deterministic,  if a configuration repeats in the protocol, the protocol is eventually periodic. Let $t$ be an upper bound for the number of configurations reachable by the machine running from the initial configuration. If we look at the machine's state at step $t + 1$ (i.e., in the $(t+1)$-th member of the protocol), it will either be in the final state (meaning the machine halted), or it will never halt, since it visited some configuration twice, by the Pigeonhole Principle.
	
	Thus, we need an auxiliary lemma estimating the number $t$. It is quite standard. We add it for completeness of the exposition.
    
	\begin{lemma}
		Let $M$ be a deterministic TM with $q$ states and  an alphabet $\Gamma$ such that $M$ works in space $S$. Then the number of distinct configurations of $M$ is at most $q \cdot S \cdot |\Gamma|^{S}$.
	\end{lemma}
	
	\begin{proof}
        Note that a configuration consists of a state, a head position, and a tape content. We have $q$ states. The tape is represented by  a word of length $S$ in alphabet $\Gamma$. Thus, there are $|\Gamma|^{S}$ possible tape contents. Finally, the head can be over any tape cell in the working zone, so there are $S$ possible positions. In total, there are $q \cdot S \cdot |\Gamma|^{S}$ possible configurations.
	\end{proof}

		Now we provide the details of the reduction.

		Let $M = (Q, \Gamma, \sigma, q_0, q_f)$ be the fixed bordered Turing machine such that the bordered $\mathrm{SBHP}_{M}$ (Problem~\ref{NewHaltingSpace}) is \PSPACE-complete.

We need to choose an appropriate problem $RE (f,\Sigma, b)$. The alphabet is chosen as:
\[
\Sigma = \{(q, \alpha) : q \in Q \cup \{\varepsilon\}; \ \alpha \in \Gamma\}.
\]
So, an alphabet symbol is a pair (a state of TM,   a symbol of TM), where we extend the state set by a void state $\varepsilon$. The symbol $b$ is $(q_f, \wend)$. 

The  function $f(x,y,z)$ is defined by  the  transition function of $M$ as follows:
		\begin{enumerate}
			\item Suppose that $y = (q, \alpha)$, where $q \neq \varepsilon$, and  $\sigma(q, \alpha) = (q', \alpha', m)$, where $m \in \{-1, 0, 1\}$ is the head movement. If $m \neq 0$ then $f(x,y,z) = (\varepsilon, \alpha')$. If $m = 0$ then $f(x,y,z) = (q', \alpha')$.
			\item Suppose that $x= (q, \alpha)$, where $q \neq \varepsilon$, and $y = (\varepsilon, \beta)$, and $\sigma(q, \alpha) = (q', \alpha', +1)$. Then $f(x,y,z) = (q', \beta)$.
			\item Suppose that $z = (q, \alpha)$, where $q \neq \varepsilon$, and $x =(\varepsilon, \beta') $, $y = (\varepsilon, \beta)$, and $\sigma(q, \alpha) = (q', \alpha', -1)$. Then $f(x,y,z) = (q', \beta)$.
			\item Otherwise, $f(x,y,z) = y$.
		\end{enumerate}
		
		The function $f$ is defined by a constant value table of size $|\Sigma|^3$, since $\Sigma$ is fixed.
		
An instance of  $\mathrm{SBHP}_{M}$ consists of the input word $w$ and the zone size $S$, given in unary.
The reduction function must output offsets, an index $x$, and the boundary conditions for the problem $RE (f,\Sigma, b)$. Our goal is to encode the protocol of $M$ by a sequence $P(0), P(1), \dots $ in the alphabet $\Sigma$ such that 
$P[kL:kL+L-1]$, where  $L = S+2$,
 encodes a $k$-th configuration in the protocol in a natural way: the second coordinates of pairs in symbols $P(j)$ define the tape content, while the non-empty first coordinate indicates the head position and the current state. So, in a configuration encoding there is exactly one symbol with non-empty first coordinate.

   Offsets are $L - 1, L, L + 1$. 
      The boundary conditions are defined by the initial configuration  $\wbeg q_0 w \underbrace{\blankTM \dots \blankTM}_{S-|w|} \wend$ of $M$ and include also the first symbol of the second configuration. It is determined by the transition function of $M$. More exactly, if $\delta(q_0,w_1) =(q', \alpha', m)$, then $P(L) = (a, \wbeg)$, where $a=q'$ in the case $m=-1$ and $a=\varepsilon$ otherwise.
      Thus, the boundary conditions form the sequence
      \[
      P[0:L] =\big( (\varepsilon, \wbeg), (q_0,w_0), (\varepsilon, w_1), 
      \dots,(\varepsilon, w_n), (\varepsilon, \blankTM), \dots, (\varepsilon, \blankTM), (\varepsilon, \wend), (a,\wbeg)\big).
      \]
 Computing the boundary conditions takes linear time in the input length.

 Thus,  the  recurrence relation is 
 \begin{equation}\label{reduction-rec}
 P(i) = f(P(i - 1 - L), P(i - L), P(i + 1 - L)).
 \end{equation}

		\begin{claim}
			The recurrence \eqref{reduction-rec} generates the protocol of operating $M$ starting from the initial configuration.
		\end{claim}
		
		\begin{proof}
			Allowing some ambiguity in terminology, we will call encodings of configurations as ``configurations''.
			
			The proof is by induction on the symbol index in the sequence.
			
			Base case: the first $L$ initial members of the recurrence form the initial configuration.
			
			Induction step: assume the sequence $P(j)$, $j \in \{0, 1 \dots, i - 1\}$, defines $\lfloor \frac{i}{L} \rfloor$ configurations of the machine's operation and some symbols of the current configuration.
			\begin{enumerate}
				\item If symbols $P(i - 1 - L), P(i - L), P(i + 1 - L)$ are part of the same configuration, then this is the previous configuration. In this case, $P(i) = f(P(i - 1 - L), P(i - L), P(i + 1 - L))$ by definition of the function $f$.
				\item If $P(i - 1 - L)$ belongs to one configuration and $P(i - L)$ to another, then $P(i) = P(i - L)$, which corresponds to the definition of $f$ in both cases: when $P(i - L) = (\varepsilon, \wbeg)$ and when $P(i - L) = (q, \wbeg)$ (due to definition of a bordered TM).
				\item Similar arguments are applicable to the case when $P(i - L)$ belongs to one configuration and $P(i + 1 - L)$ to another.\qedhere
			\end{enumerate}
		\end{proof}
		
		We choose the value $t = |Q| \cdot (S+2) \cdot |\Gamma|^{S+2} + 1$, which is strictly greater than the number of possible configurations of $M$ in space $S+2$. Thus, the machine $M$ halts if and only if it reaches the final state $q_f$ within $t$ steps. This means $M$ halts if and only if $P((t + 1) \cdot L - 1)=(q_f, \wend)$. The target index $x = (t + 1) \cdot L - 1$ can be computed in binary in polynomial time. Recall that the target symbol is $b = (q_f, \wend)$. The reduction is constructed. 
	
The second claim of the theorem follows from this reduction and Lemma~\ref{Sigma2bin}.

\subsection{Proof of Theorem \ref{result2}}\label{sec:result2-proof}
	
First, we note that the problem $SRE (f,\Sigma, b)$ is in \EXP. To compute the value $P(x)$, we can iteratively compute $P(0), P(1), \dots, P(x)$. Since $x$ is given in binary, it takes exponential time in the input length.

    To prove \EXP-hardness, we reduce the Bordered $\mathrm{EXP\text{-}BHP}_M$ (Problem \ref{NewHaltingTime}) to $SRE (f,\Sigma, b)$. We use the same fixed bordered TM $M$ as in Lemma \ref{lem:bordered_pspace}, and the same $f$, $\Sigma$, and $b$ as in the proof of Theorem \ref{result}.

    The reduction is very similar to the one for \PSPACE, but now the space bound is not given directly. Instead, we know that if the machine $M$ runs for at most $n$ steps, then the space used is at most $n$. We set the space bound $S = n$. The configuration length is $L = S+2 = n+2$. The recurrence rule $f$ remains exactly the same. The shifts $L-1, L, L+1$ are now $n+1, n+2, n+3$, which are given in binary and their lengths are bounded by the length of the binary representation of $n$.

    The boundary conditions are defined by translating the initial configuration $\wbeg q_0 w \underbrace{\blankTM \dots \blankTM}_{n-|w|} \wend$ of $M$ into boundary conditions of the recurrence. 
They consist of the initial configuration padded by  the first symbol of the second configuration, which is computed in the same manner as in the proof of Theorem~\ref{result}. The total length of the boundary conditions $L+1 = n+3$.

    \begin{claim}
    [Th.~3.1 in \cite{PolyAlgoHasSmallCircuitPAPER}]
    \label{PolyAlgoHasSmallCircuit}
		There exists a constant $c \geq 1$ such that  $f \in \DTIME(t(m)) $ implies $ f \in \SIZE(t^c(m))$, where $\SIZE(f(m))$ is the set of languages decided by circuits of size no more than $f(m)$. Here $m$ is the input length.
	\end{claim}

    To complete the proof, we need to construct a Boolean circuit $C$ that outputs the $\tau$-encoding of the $y$-th symbol of the initial data. The circuit $C$ takes an index $y$ of length $O(\log m)$, where $m$ is the input length, and outputs the corresponding symbol encoding of length $l = |\Sigma|+8$. Since the initial configuration has a simple  structure (a word $w$ of length $|w|$ followed by $n-|w|$ blanks and three extra symbols), the $y$-th symbol can be computed in time polynomial in the input size (which is $|w| + \log n$). By Claim \ref{PolyAlgoHasSmallCircuit}, there exists a polynomial-size Boolean circuit $C$ that computes this function, and it can be constructed in polynomial time.

    We choose $t = n$. The machine $M$ halts in at most $n$ steps if and only if $P((n + 1) \cdot L - 1)=(q_f, \wend)$. The target index $x = (n + 1) \cdot L - 1$ can be computed and written in binary in polynomial time.  This completes the reduction for $SRE(f, \Sigma, b)$.

    The second part of the theorem follows from Lemma \ref{SuccSigma2bin}, which states that $SRE(f, \Sigma, b)$ is polynomially reducible to $SRE(g, \{0,1\}, 1)$ for some Boolean function $g$.

	\section{A NAND recurrence evaluation problem}
	
	In this section, we consider a recurrence relation defined on a 3-dimensional grid. Let $V: \mathbb{Z}^3 \to \{0, 1\}$ be a function defined at points with integer coordinates $(x, y, z)$. The values of the function are determined by the following recurrence: 
    \begin{equation}\label{NAND-rec}
    V(x, y, z) = 
    \lnot \big(V(x - 1, y, z)\land V(x, y - 1, z)\land V(x, y, z - 1)\big).
    \end{equation}
    A \emph{layer} of height $t$ is the set $\{(x,y,z) : x+y+z= t\}$. 
The above recurrence shows that the values of $V$ at layer $t+1$ depend on the values at layer $t$. 

In the context of impartial games, \eqref{NAND-rec} corresponds to the recurrence computing values of positions in a $3$-dimensional subtraction game with $a_1=(1,0,0)$, $a_2=(0,1,0)$ and $a_3=(0,0,1)$. The boundary conditions are assumed to be arbitrary. It means that we consider a general winning condition for the game. We refer to this game as the \emph{cone  game}. 

We define the evaluation problem for this recurrence.

    \begin{prob}[The cone game evaluation]\label{prob:eval-nand}
      Input: an integer $T > 0$ given in unary and boundary conditions $V(x,y,z) \in \{0,1\}$ for all $x,y,z \le 0$ such that $x+y+z=-T$.
      
      Output: The value $V(0, 0, 0)$ computed according to the recurrence \eqref{NAND-rec}.
    \end{prob}

To study the complexity of this problem using standard complexity classes and reductions, we define the corresponding language $\CGE \subseteq \{0,1\}^*$. The definition includes padding, which is used in the reduction below.

    \begin{defi}\label{CGE-lang}
       A word $w \in \CGE$ has the form $1^T 01^{T'} 0w'$, where $|w'| = 2^{2k}$, and $k = \lceil \log_2(T+1) \rceil$, and $T\geq T'$. The bits of $w'$ encode the boundary conditions on the layer $-T'$ as follows: for any $x, y, z \le 0$ such that $x+y+z=-T'$, the value $V(x,y,z)$ is given by the bit of $w'$ at the index $x' \cdot 2^k + y'$, where $x' = -x$ and $y' = -y$. 
      The remaining bits of $w'$ (where $x' + y' > T'$) are set to $0$. The word $w$ belongs to $\CGE$ if $V(0,0,0) = 1$ according to the recurrence \eqref{NAND-rec}.
    \end{defi}
    
Note that  $\CGE\in\Ptime$ since computing $V(0,0,0)$ from the boundary conditions on the layer $-T$ requires computing $\poly(T)$ intermediate values. We do not know whether $\CGE$ is \Ptime-complete (under logspace reductions). Below we provide some lower bounds of its complexity using the succinct version $\sCGE$. 
Recall the definition of the succinct version $sA$ for a language  $A\subseteq \{0,1\}^*$. The definition is taken from    \cite{BalcazarLozanoToran1992}. The succinct version $sA$ consists of succinct representations of words from $A$. A succinct representation of a word $w$ is a pair of Boolean circuits  that take the binary representation of an index $i$ as the input and verify that $i\leq |w|$ (the first circuit) and $w_i = 1$ (the second). Here we assume that  circuits are represented  by binary words in some appropriate way. The value of a circuit should be computed in polynomial time given the representation of the circuit and the list of input variables.

    \begin{theorem}\label{PP-complete}
       \sCGE{} is \PP-hard.
    \end{theorem}

To prove  Theorem~\ref{PP-complete}, we construct a very restricted reduction from the language $\symF$ to \CGE. The language $\symF$ is a collection of binary words $w= uv$ such that $|u|= |v| = n$ and $u_k=1$, where $k$ is the number of 1s in $v$. In other words, the words from $\symF$ represent all symmetric Boolean functions: the prefix $u$ specifies a function and the suffix $v$ specifies arguments of the function.
    
Required reductions are called logtime reductions.  For these reductions we need Turing machines running in time $O(\log n)$, where $n$ is the input length. Of course, the ordinary TM cannot read the whole input in that time. So, we assume random access to the input word. A machine has a special address tape. The content of the address tape is an index of a cell of the input word.  In one step the machine can get a bit indexed by the address tape or receive the overflow message. In this way such a machine can read any input bit. Nevertheless, the number of bits read is still     $O(\log n)$. A~logtime reduction of a language $L_1$ to a language $L_2$ is a function $f\colon \{0,1\}^*\to \{0,1\}^*$ such that 
\begin{itemize}
    \item [(i)]
 there exists a function $\varphi(x,i)$ computable in logarithmic time such that $\varphi(x,i) = f(x)_i$ if $i\leq |f(x)|$, and is undefined otherwise; 
 \item[(ii)] $x\in L_1$ iff $f(x)\in L_2$.
 \end{itemize}
   
Details can be found in \cite{BalcazarLozanoToran1992}. 

Our main technical result is the following.

    \begin{theorem}\label{sym2CGE}
      There exists a logtime reduction from the language $\symF$ to $\CGE$.
    \end{theorem}

       Since the majority function is symmetric,   the conversion lemma \cite{BalcazarLozanoToran1992} and the
        connection between succinct circuit representations and leaf
        languages \cite{BorchertLozano1996} imply that \sCGE{} is \PP-hard.

\begin{rem}
    By the same reasons, since the parity function is symmetric, Theorem~\ref{sym2CGE} implies $\plusP$-hardness of \sCGE.
\end{rem}
    
	\subsection{Layer-to-Layer Mappings}

     The value $V(x,y,z)$ at a point from layer $t+2$ is determined by the values of $V(\cdot)$ at six points from layer $t$, namely 
\[
(x-2, y, z);\ (x-1,y-1, z);\ (x, y-2,z);\ 
(x, y-1, z-1); \ (x,y,z-2);\ (x-1,y,z-1).
\]   
     This dependence is defined  by a function $g: \{0,1\}^6 \to \{0,1\}$. 
    Note that solving recurrence \eqref{NAND-rec} is reduced to the computation of values at the preceding even level $z = 2t$ with constant computational overhead.
Therefore we consider the recurrence defined by the function $g$ and will analyze the corresponding evaluation problem.

\begin{figure}[H]
	\centering
	\begin{tikzpicture}[
		xscale=2.5, yscale=0.6, 
		>=stealth,
		every node/.style={font=\small},
		p/.style={circle, draw=black, inner sep=1.2pt},
		mid/.style={circle, draw=black, fill=white, inner sep=1.2pt},
		top/.style={circle, draw=black, fill=white, inner sep=1.4pt},
		deparr/.style={->, thick},
		grayarr/.style={->, thick, gray!70},
		layerline/.style={gray!60, dashed, thick}
	]
		
		\draw[layerline] (-0.8,0) -- (4.8,0) node[right] {layer $t$};
		\draw[layerline] (-0.8,4.4) -- (4.8,4.4) node[right] {layer $t+1$};
		\draw[layerline] (-0.8,7.6) -- (4.8,7.6) node[right] {layer $t+2$};
		
		\coordinate (BL) at (0,0.3);   % (x-2,y,z)
		\coordinate (BR) at (4,0.3);   % (x,y-2,z)
		\coordinate (BT) at (2,3.1);   % (x,y,z-2)
		
		\coordinate (Bmid) at (2,0.3); % (x-1,y-1,z)
		\coordinate (Lmid) at (1,1.7); % (x-1,y,z-1)
		\coordinate (Rmid) at (3,1.7); % (x,y-1,z-1)
		
		\draw[gray!40] (BL) -- (BR) -- (BT) -- cycle;
		\draw[gray!40] (BT) -- (Bmid);
		\draw[gray!40] (BL) -- (Rmid);
		\draw[gray!40] (BR) -- (Lmid);
		
		\node[p, fill=black]   (P1) at (BL) {};
		\node[below left=1pt, black] at (P1) {$(x\!-\!2,\,y,\,z)$};
		
		\node[p, fill=black]  (P3) at (BR) {};
		\node[below right=1pt, black] at (P3) {$(x,\,y\!-\!2,\,z)$};
		
		\node[p, fill=black]     (P6) at (BT) {};
		\node[above=2pt] at (P6) {$(x,\,y,\,z\!-\!2)$};
		
		\node[p, fill=black]    (P2) at (Bmid) {};
		\node[below=2pt, black] at (P2) {$(x\!-\!1,\,y\!-\!1,\,z)$};
		
		\node[p, fill=black] (P5) at (Lmid) {};
		\node[left=2pt, black] at (P5) {$(x\!-\!1,\,y,\,z\!-\!1)$};
		
		\node[p, fill=black]    (P4) at (Rmid) {};
		\node[right=2pt, black] at (P4) {$(x,\,y\!-\!1,\,z\!-\!1)$};
		
		\coordinate (ML) at (1,4.8); 
		\coordinate (MR) at (3,4.8); 
		\coordinate (MT) at (2,6.2); 
		
		\draw[gray!40] (ML) -- (MR) -- (MT) -- cycle;
		
		\node[mid] (M1) at (ML) {};
		\node[left=2pt] at (M1) {$(x\!-\!1,\,y,\,z)$};
		
		\node[mid] (M2) at (MR) {};
		\node[right=2pt] at (M2) {$(x,\,y\!-\!1,\,z)$};
		
		\node[mid] (M3) at (MT) {};
		\node[above=2pt] at (M3) {$(x,\,y,\,z\!-\!1)$};
		
		\node[top, very thick] (T) at (2,8.2) {};
		\node[above=2pt] at (T) {$(x,\,y,\,z)$};
		
		\draw[deparr, black]   (P1) -- (M1);
		\draw[deparr, black]    (P2) -- (M1);
		\draw[deparr, black] (P5) -- (M1);
		
		\draw[deparr, black]    (P2) -- (M2);
		\draw[deparr, black]  (P3) -- (M2);
		\draw[deparr, black]    (P4) -- (M2);
		
		\draw[deparr, black] (P5) -- (M3);
		\draw[deparr, black]    (P4) -- (M3);
		\draw[deparr, black]     (P6) -- (M3);
		
		\draw[deparr, black] (M1) -- (T);
		\draw[deparr, black] (M2) -- (T);
		\draw[deparr, black] (M3) -- (T);
	
	\end{tikzpicture}
	\caption{Unrolling \eqref{NAND-rec} for two steps: the value at \((x,y,z)\) on layer \(t+2\) depends on six points on layer \(t\). These six points form a triangle of side \(2\), which maps to a triangle of side \(1\) on layer \(t+1\), and then to the target point on layer \(t+2\).}
	\label{fig:dependency-six-points}
\end{figure}
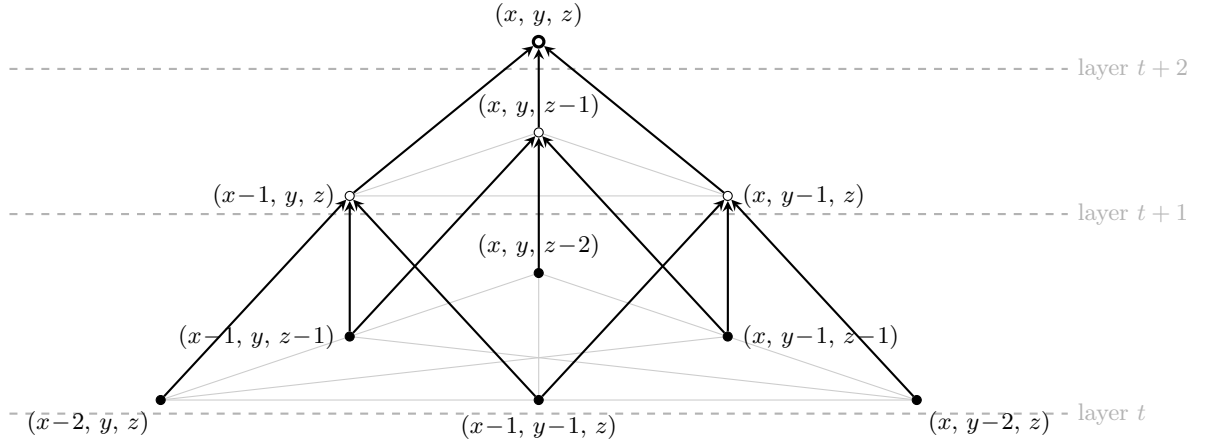

As illustrated in Figure~\ref{fig:dependency-six-points}, the recurrence \eqref{NAND-rec} can be unrolled for two steps and written as a monotone Boolean function of the six points on the lower layer:
\[
\begin{aligned}
a&:=V(x-2,y,z),&
b&:=V(x-1,y-1,z),&
c&:=V(x,y-2,z),\\
d&:=V(x,y-1,z-1),&
e&:=V(x-1,y,z-1),&
f&:=V(x,y,z-2).
\end{aligned}
\]
Thus,
\begin{equation}\label{eq:unrolled-g-six}
V(x,y,z)=g(a,b,c,d,e,f) = (a\wedge b\wedge e)\ \vee\ (b\wedge c\wedge d)\ \vee\ (e\wedge d\wedge f).
\end{equation}

This formula is exactly what Figure~\ref{fig:dependency-six-points} depicts.
The six points \(\{a,b,c,d,e,f\}\) form the lower triangle (layer \(t\)).
Inside it there are three ``small'' triples:
\[
\{a,b,e\},\qquad \{b,c,d\},\qquad \{e,d,f\},
\]
each triple feeding one of the three vertices of the middle triangle (layer \(t+1\)) via one application of \eqref{NAND-rec}.
Then the top value \(V(x,y,z)\) (layer \(t+2\)) is obtained by applying \eqref{NAND-rec} once more to these three middle vertices.
After substitution, the negations cancel by De~Morgan's laws, yielding \eqref{eq:unrolled-g-six} as an \(OR\) of three \(AND\)'s.

\medskip
\paragraph*{Rewriting in \((u,v,t)\)-coordinates.} In the analysis of \eqref{eq:unrolled-g-six}
we use the new coordinates
\[
u=-x+t,\qquad v=-y+t,\qquad 2t=x+y+z.
\]
If \((x,y,z)\) is a point on layer \(2t+2\), then the new coordinates are \((u,v,t+1)\), where  \(u=-x+(t+1)\), \(v=-y+(t+1)\).
The six lower-layer points (on layer \(t\)) from \eqref{eq:unrolled-g-six} correspond to the following new coordinates:
\[
\begin{aligned}
&(x-2,y,z )     \leftrightarrow (u+1,\,v-1,\,t), & &&\hspace*{2em}
&(x-1,y-1,z)    \leftrightarrow (u,\,v,\,t),\\
&(x,y-2,z )\leftrightarrow (u-1,\,v+1,\,t), & &&\hspace*{2em}
&(x,y-1,z-1)    \leftrightarrow (u-1,\,v,\,t),\\
&(x-1,y,z-1)    \leftrightarrow (u,\,v-1,\,t), & &&\hspace*{2em}
&(x,y,z-2)      \leftrightarrow(u-1,\,v-1,\,t).
\end{aligned}
\]
Substituting these into \eqref{eq:unrolled-g-six}, we obtain an explicit recurrence in \((u,v,t)\) coordinates:
\begin{equation}\label{tranfsormation}
\begin{aligned}
V(u,v,t+1)=&\ \big(V(u+1,v-1,t)\wedge V(u,v,t)\wedge V(u,v-1,t)\big)\\
&\ \vee\ \big(V(u,v,t)\wedge V(u-1,v+1,t)\wedge V(u-1,v,t)\big)\\
&\ \vee\ \big(V(u,v-1,t)\wedge V(u-1,v,t)\wedge V(u-1,v-1,t)\big). 
\end{aligned}
\end{equation}

Since this change of coordinates is invertible, solving recurrence \eqref{eq:unrolled-g-six} is equivalent to solving recurrence \eqref{tranfsormation}.  

The transformation \eqref{tranfsormation}  describes an evolution of configurations $U_t = \{(u,v): V(u,v,t)=1\}$ over time $t$. The corresponding map on the set of configurations is denoted by $\tau $.
In \eqref{tranfsormation}, each term of the disjunction is a conjunction of three values  (a ``small triangle'' on the previous even layer). Therefore, a point of configuration $U_{t+1}=\tau (U_t)$ is generated by  \emph{at least one} of these small triangles.
Due to monotonicity, each ``small triangle'' of $U_t$ generates another triangle in $U_{t+1}$ as shown in Fig.~\ref{fig:automaton-evolution}. Note that left bottom corner disappears and the opposite (right top) corner appears.

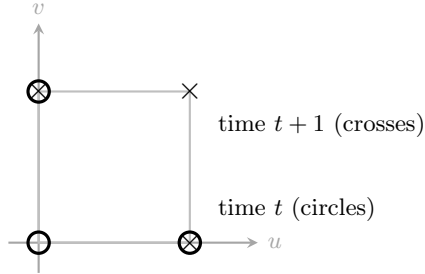
\begin{figure}[H]
	\centering 
	\begin{tikzpicture}[
		>=stealth,
		every node/.style={font=\small},
		axis/.style={->, gray!70, thick},
		sq/.style={gray!50, thick},
		old/.style={draw=black, very thick},
		new/.style={text=black, font=\bfseries\large},
	]
		\draw[axis] (-0.4,0) -- (2.9,0) node[right] {$u$};
		\draw[axis] (0,-0.4) -- (0,2.9) node[above] {$v$};

		\coordinate (BL) at (0,0);
		\coordinate (BR) at (2,0);
		\coordinate (TL) at (0,2);
		\coordinate (TR) at (2,2);
		\draw[sq] (BL) -- (BR) -- (TR) -- (TL) -- cycle;

		% Time t
		\draw[old] (BL) circle (4pt);
		\draw[old] (BR) circle (4pt);
		\draw[old] (TL) circle (4pt);
		\node[black, anchor=west] at (2.25,0.45) {time $t$ (circles)};

		% Time t+2
		\node[new] at (BR) {$\times$};
		\node[new] at (TL) {$\times$};
		\node[new] at (TR) {$\times$};
		\node[black, anchor=west] at (2.25,1.55) {time $t+1$ (crosses)};
	\end{tikzpicture}
	\caption{Evolution of a triangle in configuration $U_t$ (circles) to a triangle in configuration $U_{t + 1}$ (crosses). }
	\label{fig:automaton-evolution}
\end{figure}

\subsection{Evolution over a period of time}

\begin{lemma}[Evolution of Bounding Triangle]\label{lm:EBT}
		Assume that the configuration $U_t$ is contained in a triangle defined by inequalities $u \le u_{max}$, $v \le v_{max}$, and $u + v \ge s_{min}$. 
		
		Then $U_{t+1}$ is  contained  in a triangle $u \le u_{max}$, $v \le v_{max}$, and $u + v \ge s_{min}+1$. 
        
        In other words, the lower boundary of the configuration moves: the minimum coordinate sum of any  element of $U_{t+1}$ is strictly greater than $s_{min}$.
	\end{lemma}
	
	\begin{proof}
		From \eqref{tranfsormation} we conclude that  a non-zero value  $V(u,v,t+1)$ is possible if configuration $U_t$ contains points $(u', v')$ where $u' \ge u$ and $u' \ge u$ and at least one coordinate is strictly greater (see Fig.~\ref{fig:automaton-evolution}). The last condition implies that the minimum sum of coordinates strictly increases in $U_{t+1}$, while the first and the second conditions imply that points of $U_{t+1}$ satisfy the same inequalities     $u \le u_{max}$, $v \le v_{max}$. See Fig.~\ref{pic:triangle-evolution} for illustration.\qedhere

		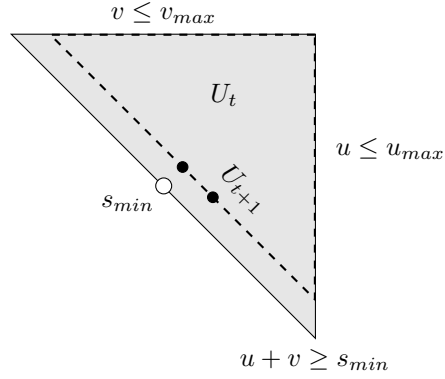
\begin{figure}[H]
			\centering
			\begin{tikzpicture}
				
				\coordinate (A) at (4,0);
				\coordinate (B) at (4,4);
				\coordinate (C) at (0,4);
				
				\coordinate (M) at (2,2);
				\coordinate (M_new) at (2.4, 2.4);
				
				\draw[thick] (A) -- (B) -- (C) -- cycle;
				
				\node at (2,4.3) {$v \le v_{max}$};
				\node at (5,2.5) {$u \le u_{max}$};
				\node at (4,-0.3) {$u + v \ge s_{min}$};
				
				\fill[gray!20] (A) -- (B) -- (C) -- cycle;
				\node at (2.8, 3.2) {$U_t$};
				
				\coordinate (A2) at (4,0.5);
				\coordinate (C2) at (0.5,4);
				\draw[dashed, thick] (B) -- (C2) -- (A2) -- cycle;
				\node[rotate=-45] at (3.05, 1.95) {$U_{t + 1}$};
				
				\draw[fill=white] (M) circle (3pt);
				\node[below left] at (M) {$s_{min}$};
				
				\draw[fill=black] (2.25, 2.25) circle (2pt);
				\draw[fill=black] (2.65, 1.85) circle (2pt);
				
			\end{tikzpicture}
			\caption{Evolution of the bounding triangle. The configuration shrinks and  the diagonal $u+v = s_{\min}$ is empty in the new configuration.}\label{pic:triangle-evolution}
		\end{figure}
	\end{proof}

\begin{rem}
    Applying Lemma~\ref{lm:EBT} subsequently, we see that a bounded configuration disappears after at most $(u_{max}+v_{max}) -s_{\min}+1 $ time steps. 
\end{rem}

\begin{defi}[The ``+1'' property]
    A  configuration $U$  satisfies the ``+1'' property during a   period of time $\Delta t$ if in a sequence 
    of configurations
    \[
    U= U_t \xrightarrow{\tau} U_{t+1}\xrightarrow{\tau} \dots\xrightarrow{\tau} U_{t+\Delta t}
    \]
    the minimum sum of coordinates of configuration points increases by exactly one over each step.  
\end{defi}

\subsection{Special Configurations: Strips and Corners}

Let us consider configurations of a special form: a \emph{vertical strip}, a \emph{horizontal strip} and a \emph{corner} (which connects a vertical and a horizontal strip). These \emph{gadgets} serve as the building blocks for our  construction.

\begin{figure}[H]
    \centering
    \begin{tikzpicture}[scale=0.5, every node/.style={font=\small}]
        \tikzstyle{grid}=[gray!30, very thin]
        \tikzstyle{cell}=[circle, fill=black, inner sep=1.5pt]
        
        % --- Vertical Strip ---
        \begin{scope}[shift={(0,0)}]
            \draw[grid] (-1,-1) grid (3,5);

            \foreach \y in {1,2,3,4} {
                \node[cell] at (0,\y) {};
            }
            \foreach \y in {1,2,3,4} {
                \node[cell] at (1,\y) {};
            }
            \node[circle, fill=white, draw=black, inner sep=1.5pt] (graynode_v) at (1,0) {};
            \draw[thick] (graynode_v.north west) -- (graynode_v.south east);
            \draw[thick] (graynode_v.north east) -- (graynode_v.south west);
            \node[below] at (1,-1.2) {Vertical strip};
            \draw[->, thick, black, dashed] (0.5, -0.2) -- (0.5, 4.2);
        \end{scope}

        % --- Horizontal Strip ---
        \begin{scope}[shift={(6,0)}]
            \draw[grid] (-1,-1) grid (5,3);
            \node[cell] at (0,1) {};
            \node[circle, fill=white, draw=black, inner sep=1.5pt] (graynode) at (1,0) {};
            \draw[thick] (graynode.north west) -- (graynode.south east);
            \draw[thick] (graynode.north east) -- (graynode.south west);
            \foreach \x in {1,2,3,4} { 
                \node[cell] at (\x,1) {};
            }
            \foreach \x in {2,3,4} { 
                \node[cell] at (\x,0) {};
            }
            \node[below] at (2,-1.2) {Horizontal strip};
            \draw[->, thick, black, dashed] (-0.2, 0.5) -- (4.2, 0.5);
        \end{scope}

        % --- Corner ---
        \begin{scope}[shift={(13,0)}]
            \draw[grid] (-1,-1) grid (5,5);
            \node[cell] at (0,1) {};
            \node[cell] at (1,1) {};
            \node[circle, fill=white, draw=black, inner sep=1.5pt] (graynode_c) at (1,0) {};
            \draw[thick] (graynode_c.north west) -- (graynode_c.south east);
            \draw[thick] (graynode_c.north east) -- (graynode_c.south west);
            \foreach \x in {0,1,2,3,4} \foreach \y in {2,3} {
                \node[cell] at (\x,\y) {};
            }
            \node[below] at (2,-1.2) {Corner};
            
            \draw[->, thick, black, dashed] (0.5, -0.5) -- (0.5, 2.5) -- (4.5, 2.5);
        \end{scope}
    \end{tikzpicture}
    \caption{The three gadgets: a vertical strip (left), a horizontal strip (center), and a corner connecting vertical and horizontal strips. Points marked with a cross indicate reference points of the gadgets.}
    \label{fig:special_configs}
\end{figure}
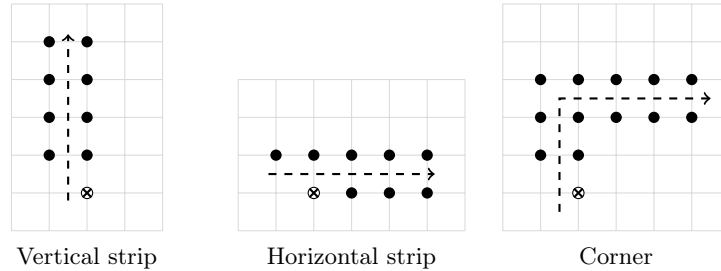

Informally,  evolution of these configurations can be viewed as a signal propagation:

\begin{itemize}
    \item A \textbf{vertical strip} (of width 2)  propagates a signal vertically.
    \item A \textbf{horizontal strip} (of width 2) propagates a signal horizontally.
    \item A \textbf{corner} connects a vertical strip and a horizontal strip, allowing a signal to turn.
\end{itemize}

The \emph{height} of a configuration $U$ is the maximal difference in $v$-coordinates of its points, i.e.,
$\max(v_2-v_1 : (u_1, v_1)\in U; (u_2,v_2)\in U)$.
Similarly, the \emph{width} of a 
configuration $U$ is the maximal difference in $u$-coordinates of its points, i.e.,
$\max(u_2-u_1 : (u_1, v_1)\in U; (u_2,v_2)\in U)$.

	\begin{lemma}[Evolution of strips]\label{lem:strip-stability} 
		Applying transformation \eqref{tranfsormation} to a vertical strip of width $1$ and height $h>1$, gives a vertical strip of height $h-1$ and preserves $u$-coordinate of the reference point.
		
		Similarly, applying transformation \eqref{tranfsormation} to a horizontal strip of height $1$ and  width $w>1$,  gives a horizontal strip of width $w-1$ and preserves $v$-coordinate of the reference point. 
	\end{lemma}
\begin{proof}
Comparing Figs.~\ref{fig:automaton-evolution}
and~\ref{fig:special_configs} we conclude that the points with minimum sum of coordinates disappear, while all other points are preserved in the next configuration.
\end{proof}

	\begin{lemma}
    \label{lem:plus-one-conservation}
     The ``$+1$'' property holds for
     \begin{itemize}
         \item a vertical strip of height $h$
          during the period of time $\Delta t = h-1$;
          \item  a horizontal strip of width $w$  during the period of time $\Delta t = w-1$;
          \item a corner of height $h$ and width $w$ during the period of time $\Delta t = h+w-2$.
     \end{itemize}
 	\end{lemma}

	\begin{proof} Follows from Lemma~\ref{lem:strip-stability} and the following observation. Note that a horizontal strip of width $1$ is also a vertical strip of height $1$. This configuration disappears at one step.  The minimum of sum of coordinates increases by $1$ during this step. 
	\end{proof}

    \subsection{Independent Configurations}

    Configurations $U'$ and $U''$ are called \emph{independent} if $\tau(U'\cup U'') =\tau (U') \cup \tau (U'') $. 
    
    First, we note a causality property: points of the configuration with  $u+v > b$ do not affect the change of the configuration at points with $u+v \leq b$. 
    
    We will use the following sufficient condition of independence within a finite strip $a \leq u+v \leq b$.
    
    A \emph{monotone path} is a sequence of points on the plane where each next point differs from the previous one by increasing exactly one coordinate by 1.

    	\begin{lemma}[Independence in a finite strip]\label{lem:independence}
        Consider a finite strip $S = \{(u,v): a \leq u+v \leq b\}$.
            If $U'$ and $U''$ are subsets of $S$
        separated by a monotone path within the strip, then $U'$ and $U''$ are independent.
    	\end{lemma}
    	\begin{proof}
    Assume that the monotone path $M$ divides the strip $S = \{(u,v): a \leq u+v \leq b\}$ into two parts $L'$ and $L''$ such that $U'\subseteq L'$ and $U''\subseteq L''$.
    It is sufficient to prove that $\tau(U')\in L'$ and $\tau(U'')\in L''$.
        Inspecting Fig.~\ref{fig:automaton-evolution}, we conclude that 
        \begin{equation}\label{tauU}
    \text{if } (a,b)\in \tau (U)\ \text{then}\ U\cap \{(u, b): a-1 \leq u\leq a\}\ne \es\ \text{and}\  
    U\cap \{(a, v): b-1\leq v\leq b\}\ne \es. 
    \end{equation}
    Suppose for contradiction that $(a,b) \in \tau (U')\cap (M \cup L'')$. Since $M$ is monotone, at least one point $(a-1,b)$ or $(a,b-1)$ is also in $M \cup L''$. It follows from \eqref{tauU} that $(a,b)\notin \tau(U')$ and we come to a contradiction. The case  $(a,b) \in \tau (U'')\cap (M\cup L')$ is considered similarly.
    \end{proof} 
	
	\subsection{The logtime reduction of \boldmath$\symF$ to $\CGE$}

Given a word $(\eta_1,\dots, \eta_n; \xi_1,\dots, \xi_n)$, where $\eta$ determines a symmetric Boolean function and $\xi$ is an assignment of the variables of the function, we will construct the input data for the cone game evaluation problem (Problem \ref{prob:eval-nand}): the boundary conditions and the value $T$. As an intermediate step we define  a configuration $I(\eta,\xi)$. Its informal structure is shown in Fig.~\ref{gadget2}. The configuration is a union of horizontal strips,  vertical strips and corners. They are depicted in Fig.~\ref{gadget2} by horizontal, vertical lines and lines at an angle $45^{\circ}$ respectively.

        Bold points in the figure are reference points of the gadgets.  
    The  point $\alpha = (0,-4)$ is the reference point of a vertical strip of width 1 and height 4. The reference points $r_{ij}$ in the trapezoid are indexed by pairs  $(i,j) $, where $0\leq i\leq n$ and $0\leq j\leq n+i$. Their $(u,v)$ coordinates are
    \[
    r_{ij} = (3i-3j, 3i+3j). 
    \]
If $\xi_{i+1}=1$ then every point $r_{ij}$ is the reference point of a vertical strip of height $6$ and width $1$. Otherwise, $r_{ij}$ is the reference point of a corner of height $4$ and width~$3$.  

\begin{figure}[h!]
\centering
		\begin{tikzpicture}[scale=1.5, every node/.style={font=\small}]
			\def\s{3}
			\def\shift{0.4}
			\def\bigshift{0.8}
			
			\coordinate (A) at (0,0);
			\node[below right=6pt and 0.5pt] at (A) {};
			\coordinate (D) at ({-\s/sqrt(2)}, {\s/sqrt(2)});
			\coordinate (C) at (0, {\s*sqrt(2)});
			\coordinate (B) at ({\s/sqrt(2)}, {\s/sqrt(2)});
			
			\coordinate (Cext) at ($ (B)!1.5!(C) $);
			
			\coordinate (D_vert) at ($ (D) + (0, 5) $);
			\coordinate (Topm) at (intersection of D--D_vert and B--Cext);
			
			\coordinate (TopRight) at (Cext -| B);
			
			\draw[thin] (A) -- (B) -- (Topm);
			\draw[thin] (A) -- (D);
			\draw[thin, dotted] (D) -- (C);
			
			\foreach \t/\lbl/\N in {0/0/Top0,0.125/2/Top2,0.25/4/Top4} {
				\coordinate (\N) at ($ (B)!\t!(C) $);
				\fill[black] (\N) circle (1.2pt);
			}
			
			\fill[black] (Topm) circle (1.2pt);
			\node[left=1.2pt] at (Topm) {$m'$};
			
			\coordinate (Bvert) at ($ (B) + (0, 3.5) $);
			\draw[very thick, black] (B) -- (Bvert);
            \fill[white] (Bvert) circle (1.2pt);
            \draw[black, thick] (Bvert) circle (1.2pt);
			\node[right=1.2pt, black] at (Bvert) {$\phi$};
			
			\foreach \N in {Top0,Top2,Top4} {
				\coordinate (H) at (\N -| B);
				\draw[very thick, black] (\N) -- (H);
			}
			\coordinate (C_H) at (C -| B);
			\draw[very thick, black] (C) -- (C_H);
			\fill[black] (C) circle (1.2pt);
			
			\coordinate (P0) at (A);
			\coordinate (P1) at ($ (A)!0.07!(D) $);
			\coordinate (P2) at ($ (A)!0.15!(D) $);
			\coordinate (Pm) at (D);
			
			\coordinate (Aup) at ($ (A) + (0,\shift) $);
			\coordinate (Dup) at ($ (D) + (0,\shift) $);
			\coordinate (P1up) at ($ (P1) + (0,\shift) $);
			\coordinate (P2up) at ($ (P2) + (0,\shift) $);
			
			\coordinate (Aup2) at ($ (Aup) + ({\shift/2}, {\shift/2}) $);
			\coordinate (Dup2) at ($ (Dup) + ({\shift/2}, {\shift/2}) $);
			\coordinate (P1up2) at ($ (P1up) + ({\shift/2}, {\shift/2}) $);
			\coordinate (P2up2) at ($ (P2up) + ({\shift/2}, {\shift/2}) $);
			
			\coordinate (DupInt) at (intersection of Aup--Dup and Pm--Topm);
			\coordinate (Dup2Int) at (intersection of Aup2--Dup2 and Pm--Topm);
			
			\coordinate (AupInt) at (intersection of Aup--Dup and A--B);
			\coordinate (Aup2Int) at (intersection of Aup2--Dup2 and A--B);
			
			\draw[thin, black] (DupInt) -- (AupInt);
			\node[black, below right] at (AupInt) {$x_1 = 1$};
			\draw[thin, black] (Dup2Int) -- (Aup2Int);
			\node[black, below right] at (Aup2Int) {$x_2 = 0$};
			\node[black, below right=1.2pt] at (B) {$x_n$};
			
			\foreach \P/\lbl in {P0/0, P1/1, P2/2, Pm/m}
			{
				\fill (\P) circle (1.2pt);
				\node[below left=1.2pt] at (\P) {$\lbl$};
			}
			\foreach \P in {Aup, P1up, P2up, Dup, Aup2, P1up2, P2up2, Dup2}
			\fill[black] (\P) circle (1.2pt);
			
			\draw[very thick, black] (P0) -- (Aup);
			\draw[very thick, black] (P1) -- (P1up);
			\draw[very thick, black] (P2) -- (P2up);
			\draw[thin, black] (Pm) -- (Topm);
			\draw[very thick, black] (Pm) -- (DupInt);
			\draw[very thick, black] (Aup) -- (Aup2);
			\draw[very thick, black] (P1up) -- (P1up2);
			\draw[very thick, black] (P2up) -- (P2up2);
			\draw[very thick, black] (DupInt) -- (Dup2);
			
			\coordinate (ALPHA) at ($ (A) + (0, -0.7) $);
			\draw[very thick, black] (A) -- (ALPHA);
			\fill[black] (ALPHA) circle (1.2pt);
			\node[right=1.2pt, black] at (ALPHA) {$\alpha$};
		\end{tikzpicture} 
		\caption{The configuration $I(\eta, \xi)$.
        }\label{gadget2}
\end{figure}
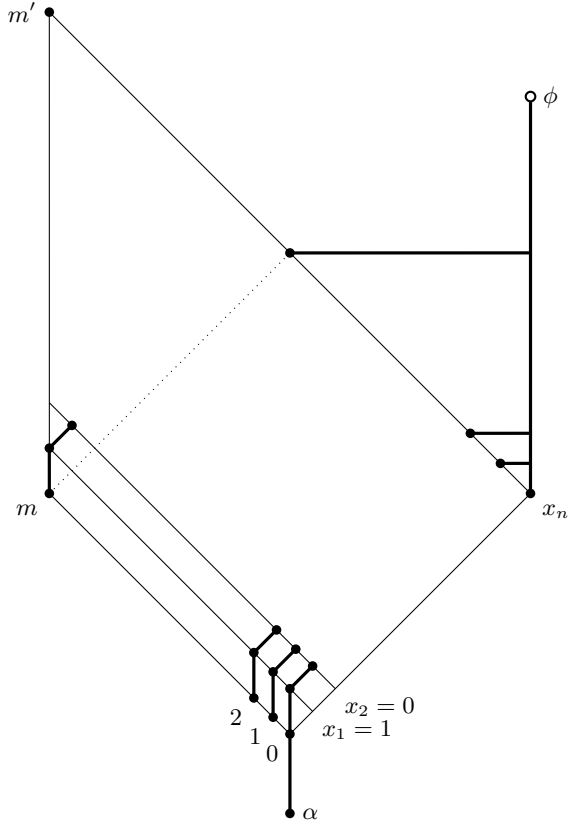

We are going to reduce  $\symF$ to $\CGE$. For an input $(\eta,\xi) \in \{0,1\}^{2n}$, the condition $(\eta,\xi) \in \symF$ is equivalent to $|\xi| \in S_{\eta}$, where $S_{\eta} = \{j \mid \eta_j=1\}$ and $|\xi|$ is the Hamming weight of $\xi$. 
The point $r_{n,0}$ is the reference point of a vertical strip of height $3n+4$ and width $1$. If $j\in S_{\eta}$ then a point $r_{n,j}$ is the reference  point of a horizontal strip of height $1$ and width $3j$.

Let $T_{\text{gadget}} = 2\cdot((3n+6n+4)-(0-4)) = 18n+16$.
The point $\phi = (\phi_1, \phi_2) = (3n, 6n+4)$ is the target point of the configuration.         
Let $I(\eta,\xi)_t$ be a configuration produced from $I(\eta,\xi)$ after $t$ steps of evolution.

\begin{lemma}\label{I(x)=f(x)}
    $\phi\in I(\eta,\xi)_{T_{\text{\textup{gadget}}}/2}$ iff $(\eta,\xi) \in \symF$.
\end{lemma}

\begin{proof}
	The proof relies on tracking the minimum sum of coordinates of the evolving configurations $I(\eta, \xi)_t$. Let $s_{\min}(U)$ denote the minimum sum of coordinates $u+v$ for points $(u,v) \in U$.

	For the initial configuration $I(\eta,\xi) = I(\eta,\xi)_0$ we have $s_{\min}(I(\eta,\xi)_0) = -4$. 
Thus,  $s_{\min}(U_{t}) \ge -4+t $ due to  Lemma~\ref{lm:EBT}.   
    So, the minimum coordinate sum of $I(\eta, \xi)_{T_{\text{gadget}}/2}$ is at least $-4 + 9n + 8 = 9n+4$. 
	
    Notice that the target point $\phi = (3n, 6n+4)$ has exactly this coordinate sum: $3n + 6n + 4 = 9n+4$. Thus, $\phi \in I(\eta,\xi)_{T_{\text{gadget}}/2}$ if and only if the  configuration $I(\eta,\xi)$ itself  satisfies the ``$+1$'' property.

	By Lemma~\ref{lem:plus-one-conservation}, the gadgets (strips and corners) satisfy the ``$+1$'' property, meaning $s_{\min}(U_{t+1}) = s_{\min}(U_t) + 1$ as long as the configuration has not disappeared.   
    For each $0\leq i\leq n-1$, subconfigurations of $I(\eta,\xi)_t$ in a diagonal strip $\{(u,v) : 6i\leq u+v< 6(i+1)\}$ consist of vertical strips and corners. 
    The evolution of these gadgets in the diagonal strip proceeds independently due to Lemma~\ref{lem:independence}. The gadgets in the strip disappear at different moments of time. 
    Note that gadgets with offset points $r_{ij}$ and $r_{(i+1) j}$ overlap on the line $u+v= 6(i+1)$. Using this fact, it is easy to check by induction on $i$ that   a gadget with the offset point $r_{ij}$ disappears at moment $t_{ij}$, where 
    \[
    t_{ij} = \left\{\begin{aligned}
    &4+6(i+1), &&\text{if}\ \sum_{m=1}^{i+1} \xi_m = j, \\
    & 6(i+1), &&\text{otherwise.}
    \end{aligned}\right.
    \]
    
It means that $t_{n-1, |\xi|} = 4+6n$ and $t_{n-1, j } = 6n$ for $j\ne |\xi|$.

Assume that  $|\xi| \in S_{\eta}$. 
By construction, $I(\eta,\xi)$ contains a horizontal strip starting at $r_{n, |\xi|} = (3n-3|\xi|, 3n+3|\xi|)$ of width $3|\xi|$. The strip overlaps with the gadget with the offset point $r_{(n-1), |\xi|}$ on the line $u+v= 6n$. 
    Thus,  $s_{\min}(I(\eta,\xi)_{6n+4}) = 6n $ if $|\xi| \in S_{\eta}$. 

If $|\xi| \notin S_{\eta}$, then, in the half-plane $u+v>6n$, there are no points in $I(\eta, \xi)$ neighbouring the gadget with the offset point $r_{(n-1), |\xi|}$ .  It implies that   $s_{\min}(I(\eta,\xi)_{6n+4}) \geq  6n+4 $. 

The configuration $I(\eta,\xi)_{6n+4}$ is formed by the vertical gadget with offset point  $r_{n,0}$ and a family of separated horizontal gadgets. This configuration satisfies the ``$+1$'' property due to Lemmas~\ref{lem:plus-one-conservation} and \ref{lem:independence} and the following fact: the merge of a horizontal and a vertical strip satisfies the ``$+1$'' property. This fact is proved similarly to the proof of Lemma~\ref{lem:plus-one-conservation}.

Therefore, if 	$(\eta,\xi) \in \symF$ then  $s_{\min}(I(\eta,\xi)_{6n+4}) = 6n $ and $\phi\in I(\eta,\xi)_{T_{\text{gadget}}/2}$. Otherwise, $I(\eta,\xi)_{T_{\text{gadget}}/2}=\es$.
\end{proof}

Since $T_{\text{gadget}}$ is even, Lemma~\ref{I(x)=f(x)} provides a reduction of deciding $(\eta,\xi) \in \symF$ to solving recurrence~\eqref{tranfsormation} at the point $\phi_1 =(3n, 6n+4, 9n+8)$ with boundary conditions specified by the configuration $\{(u,v,0): (u,v)\in I(\eta, \xi)\}$. This problem is easily reduced to the language $\CGE$ (Definition~\ref{CGE-lang}). 

Let us show why this reduction is computable in $\DLOGTIME$.
Given an index $m$ of a bit in the target word $w = 1^T 0 1^{T'}0 w'$, the $\DLOGTIME$ machine must compute $w_m$. The length $n$ of  $\xi$ can be determined in $O(\log n)$ time (via binary search for the end of the input and dividing by~2). 

Note that additions and multiplication/divisions by constants can be made in logarithmic time. So, we  restrict actions of the reducing algorithm to operations of this sort. 
To avoid   multiplications of large integers, we pad the time parameter to a power of two. The machine computes $l = \lceil \log_2 n \rceil$ (by finding the most significant bit of $n$). Then it sets $T = 2^{l+5}$ and $T' = T_{\text{gadget}} = 18n+16$. Notice that $T = 32 \cdot 2^l \ge 32n > 18n+16 = T_{\text{gadget}}$. The length of the word $w$ is exactly $|w| = T + 1 +T'+1 + 2^{2k}$, where $k = \lceil \log_2(T+1) \rceil = l+6$. Since $T$ and $2^{2k}$ are powers of two, the machine can easily check if $m \le |w|$ by comparing $m$ with $2^{2l+12} + 2^{l+5} + 18n+16+ 2$, which can be done in $O(\log n)$ time.

If $m \le T$, the machine outputs $1$. If $m = T+1$, it outputs $0$. 
If $ T+1<m\leq T+1+T'$, the machine again outputs $1$. If $m = T+1+T'+1$, it outputs $0$. If $m> T+T'+2$, 
the machine considers the shifted index $m' = m - (T+T'+3)$. Since $|w'| = 2^{2k}$, the binary representation of $m'$ has at most $2k$ bits. The machine extracts the upper $k$ bits as $x'$ and the lower $k$ bits as $y'$. This extraction is a trivial operation in $\DLOGTIME$ (just reading specific bits from the address tape).

Then, the machine checks if $x' + y' \le T'$. If not, it outputs $0$ (these are padding bits). Otherwise, it computes the original coordinates $x = -x'$ and $y = -y'$. 

To map these coordinates to the gadget coordinates $(u,v)$, we align the target point $\phi$ with the apex $(0,0,0)$ and shift the gadget down to the layer $-T'$. This gives the linear transformation:
$u = -x - 6n - 8 = x' - 6n - 8$ and $v = -y - 3n - 4 = y' - 3n - 4$.
Computing $6n$ and $3n$ requires only bit shifts and additions, which can be done in $\DLOGTIME$.

Finally, the machine checks whether $(u,v) \in I(\eta,\xi)$ using the following steps:
\begin{enumerate}
    \item \textbf{Grid localization.} The initial configuration $I(\eta,\xi)$ is constructed by placing basic gadgets at reference points $r_{ij} = (3i-3j, 3i+3j)$. Since the gadgets have a bounded constant size, any point $(u,v)$ can belong to at most one gadget. We uniquely determine the indices of this candidate gadget by inverting the coordinate mapping: $i = \lfloor (u+v)/6 \rfloor$ and $j = \lfloor (v-u)/6 \rfloor$. 
    \item \textbf{Constant part of the configuration.}
 If $i>n$ then the point belongs to the constant part of $I(\eta,\xi)$, which does not depend on $\xi$. If $j=0$ then
 $(u,v)\in I(\eta,\xi)$ iff $3n-1\leq u\leq 3n $ and $3n\leq v\leq 6n+4$.
 Otherwise, set $j' =\lfloor (2v- 6n)/6 \rfloor$.
 In this case  $(u,v)\in I(\eta,\xi)$ iff $j'\in S_{\eta}$, and $0<j'\leq n $, and  $ v \bmod 3 \in\{0,1\} $.

On the following steps we assume that $0\leq i\leq n$.
    
    \item \textbf{Relative coordinates.}  To see if $(u,v)$ is actually a filled cell within the gadget at $r_{ij}$, we compute its local offset: $\Delta u = u - (3i-3j)$ and $\Delta v = v - (3i+3j)$. If these offsets fall outside the constant bounding box of our basic gadgets, the point is definitely empty.
    \item \textbf{Input query.} The exact shape of the gadget at $r_{ij}$ is dictated by $\xi$. As defined in our construction, if $i < n$, the gadget at $r_{ij}$ is either a vertical strip or a corner depending entirely on the value of the input bit $\xi_{i+1}$. 
    Thus, we need to get $\xi_{i+1}$ from the input tape. Writing the $O(\log n)$-bit index $i+1$ to the address tape to query the bit takes $O(\log n)$ time.
    \item \textbf{Local pattern matching.} Knowing the exact gadget type from the input query and the relative coordinates $(\Delta u, \Delta v)$, we simply check if this local position is ``filled'' in this specific gadget. This is a finite table lookup (a constant-size Boolean condition) evaluated in $O(1)$ time.
\end{enumerate}

All  operations of the algorithm are performed in  $O(\log n)$ time. Thus, the entire reduction is performed in $O(\log n)$ time.

It completes the proof of Theorem~\ref{sym2CGE}.
         
	\section{Open Problems and Future Work}

        Our findings open several avenues for future research. First, we established tight complexity bounds for general recurrent sequences. The goal is to extend the bounds to NAND recurrent sequences. In terms of games, it would mean \PSPACE- or \EXP-completeness of subtraction games with arbitrary winning conditions. Our techniques have no clear extension to this case. Some new ideas are required.

Second, while we established \PP-hardness of \sCGE, the exact complexity of languages \CGE{} and \sCGE{} remains open. As for \sCGE, it would be interesting to determine whether it is \Ptime-complete, or if it belongs to a smaller class such as $\TC^0$. 

A third and most important goal is to generalize the results to multidimensional subtraction games with non-negative offsets for normal and/or mis\`ere winning conditions. For this purpose, one can consider recurrences with the same simple set of offsets but with more complicated boundary conditions.

\paragraph*{Acknowledgements.} This article is an output of a research project (HSE-BR-2025-24) implemented as part of the Basic Research Program at HSE University.

	\printbibliography[heading=bibintoc]

\end{document}